\documentclass[letterpaper, 10 pt, conference]{ieeeconf}  
\usepackage{graphicx}
\usepackage{hyperref}
\usepackage[dvipsnames]{xcolor}
\usepackage{cite}
\usepackage{layouts}
\usepackage{algorithm}
\usepackage{algpseudocode}

\IEEEoverridecommandlockouts                              

\usepackage{amsmath,amsfonts,amssymb,amsthm}
\newtheorem{prop}{Proposition}
\newtheorem{thm}{Theorem}
\newtheorem{corr}{Corollary}
\newtheorem{lem}{Lemma}
\newtheorem{prob}{Problem}

\newtheorem{defn}{Definition}

\newcommand{\ind}[1]{{1\!\!1}\left\{#1\right\}}
\newcommand{\xlo}{\underline{x}}
\newcommand{\xhi}{\overline{x}}

\newcommand{\R}{\mathbb{R}}

\newcommand{\N}{\mathbb{N}}

\title{Forward Reachability for Discrete-Time Nonlinear Stochastic Systems via Mixed-Monotonicity and Stochastic Order}

\author{Vignesh Sivaramakrishnan, Rosalyn A. Devonport, Murat Arcak, and Meeko M.K. Oishi
        \thanks{This material is based upon research supported by the Air Force Research Lab (AFRL) under agreement number FA9550-23-1-0646, and in collaboration with Verus Research under agreement number FA9453-23-C-A025.  The U.S. Government is authorized to reproduce and distribute reprints for Governmental purposes notwithstanding any copyright notation thereon.  The views and conclusions contained herein are those of the authors and should not be interpreted as necessarily representing the official policies or endorsements, either expressed or implied, of Air Force Research Laboratory (AFRL) and or the U.S. Government.}
        \thanks{This work is funded in part by NSF grant CNS-2111688.}
        \thanks{Vignesh Sivaramakrishnan is a graduate student with Electrical and Computer Engineering, University of New Mexico, Albuquerque, NM.\newline Email: \tt{vigsiv@unm.edu}.} 
        \thanks{Rosalyn Alex Devonport is a postdoctoral scholar with Electrical and Computer Engineering, University of New Mexico, Albuquerque, NM.\newline Email: \tt{devonport@unm.edu}.} 
        \thanks{Murat Arcak is a Professor with Electrical Engineering and Computer Sciences, University of California, Berkeley, CA.\newline Email: \tt{arcak@berkeley.edu}.}
        \thanks{Meeko M.K. Oishi is a Professor with Electrical and Computer Engineering, University of New Mexico, Albuquerque, NM.\newline Email: \tt{oishi@unm.edu}.}
}

\begin{document}
\maketitle

\begin{abstract}
We present a method to overapproximate forward stochastic reach sets of discrete-time, stochastic nonlinear systems with interval geometry.
This is made possible by extending the theory of mixed-monotone systems to incorporate stochastic orders, and a concentration inequality result that lower-bounds the probability the state resides within an interval through a monotone mapping.  
Then, we present an algorithm to compute the overapproximations of forward reachable set and the probability the state resides within it. 
We present our approach on two aerospace examples to show its efficacy.
\end{abstract}

\section{Introduction}

Ensuring whether a system will perform as one desires is a crucial step in its design and operation.
For example, in aerospace, accurately quantifying the future positions and orientations of a spacecraft over long time horizons is crucial for mission success~\cite{Holzinger2021Cislunar}.  
Forward reachability provides a mathematical formalism to quantify the set of states a system will reach from an initial set of states.
Specifically, \textit{forward stochastic reachability} finds not only the set of states a stochastic system will reach, but also provides the probability its states lie within it.

Forward stochastic reachability for discrete time linear systems has a rich literature that we may divide according to the regularity assumptions used to compute the stochastic reachable set. 
The most popular collection of assumptions comprises linearity of the system and boundedness of the uncertainty, either as a random variable~\cite{kvasnica2015reachability,kurzhanskiy2006ellipsoidal,girard2005reachability} or as a characteristic function~\cite{vinod2017forward,vinod2020probabilistic}; Other works replace linearity with a weaker assumption of nonlinear structure~\cite{sankaranarayanan2020reasoning,stankovic1996taylor,estrada1993taylor, xue2017just, xue2021reach}.

Data-driven approaches~\cite{thorpe2021learning,lew2021sampling,devonport2020estimating,devonport2023data,hashemi2023data} employ the weakest possible assumptions---typically no more than measurability;
however, to ascertain a strong guarantee on the value of the state requires a significant number of samples, without a clear notion whether the data-driven representation is an over or under approximation. 

Our interests lie in directly obtaining the forward reach set in a tractable manner in nonlinear systems with order-preserving dynamics. The archetype of such systems is the class of \emph{monotone} systems~\cite{angeli2003monotone} prominent in cooperative multi-agent systems, biological systems, and traffic networks~\cite{coogan2016mixed}. 
The monotonicity property is particularly useful for computing interval approximations to forward reachable sets, as the tightest interval over-approximation to a reachable set can be computed with two dynamical simulations. Many extensions exist available for systems that do not satisfy the true order-preserving property; there is \emph{mixed monotonicity}~\cite{enciso2006nonmonotone,  angeli2014small, coogan2016mixed} for systems that can be decomposed into increasing and decreasing parts, and \emph{sensitivity}~\cite{meyer2018sampled,meyer2020ifac} for extension of mixed-monotonicity into continuous-time. 
The applications we consider in this paper focus on the case of mixed monotonicity. Mixed monotone methods provide mode conservatiev overapproximations than either pure monotone methods or sensitivity methods; but because of the relative restrictiveness of the pure monotonicity assumption, and the difficulties involved in computing the sensitivity functions, mixed monotone methods are the most pragmatic class of methods out of the three. Thus, in this paper \emph{we elect to make a trade-off that favors extreme computational efficiency, wide applicability, and ease of construction of algorithms at the cost of introducing conservatism into the reachable set over-approximation}.
It is beyond the scope of this paper to make a comparison of the potential of all three methods for stochastic generalization: for a general overview of interval reachability, with a particular emphasis on order-preserving properties and their generalizations, see~\cite{meyer2019tira,meyer2021interval,coogan2020mixed}. 

Order-preservation properties are predominantly studied in the non-stochastic setting. In order to extend this useful and efficient class of reachability methods to the setting of stochastic reachability, it would be necessary to extend the notion of order preservation to one that encompasses orderings among random variables.
\textit{
The main contribution of this paper is to provide such an extension; that is, a theory for mixed-monotone dynamics with stochastic orders that enables us to estimate forward stochastic reach sets of discrete-time, nonlinear stochastic systems using their mixed-monotone decompositions.
}
\begin{figure}[t!]
    \centering
    \includegraphics[width=\linewidth]{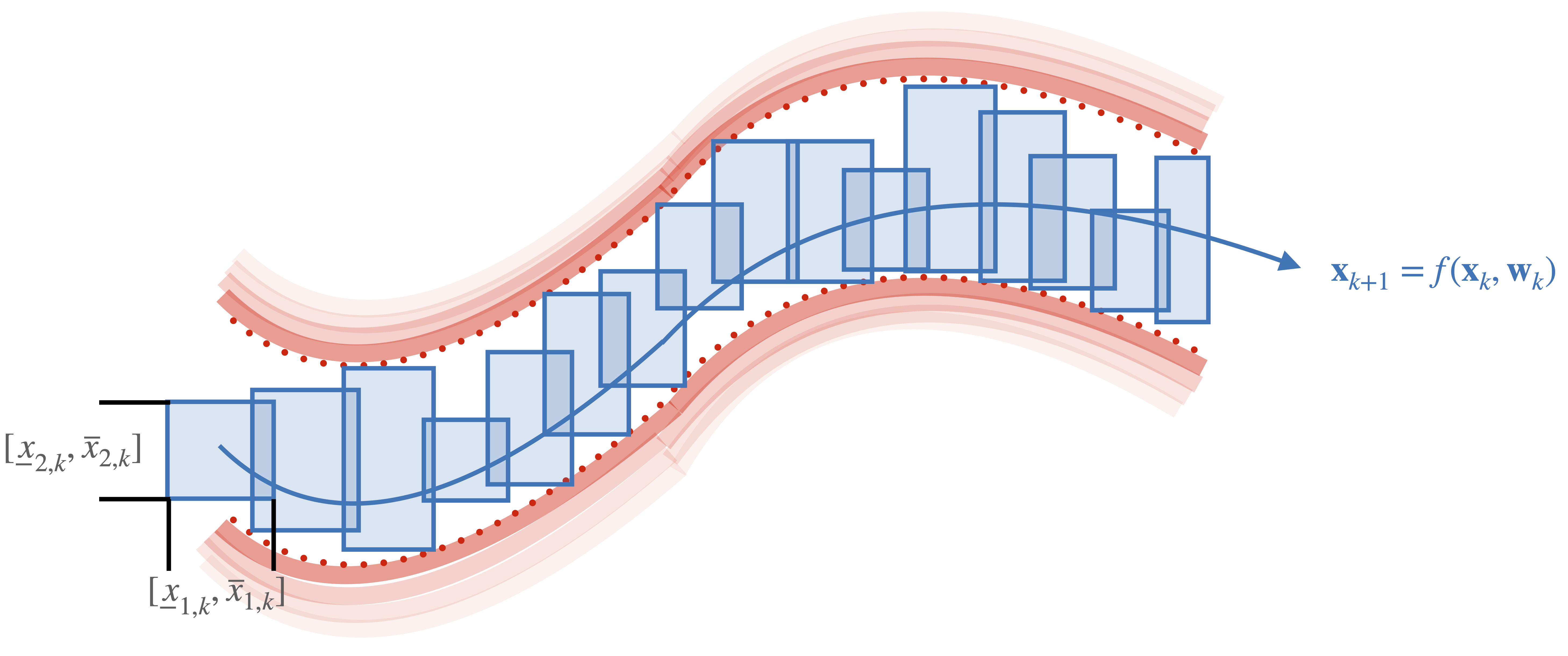}
    \caption{Our approach combines the benefits of mixed monotonicity via the concentration inequalities we derive by establishing partial order of mixed-monotone decomposition of discrete-time stochastic systems via stochastic order.
    This allows us to determine the probability (denoted in red) the states of the system lie within the overapproximation of forward reachable sets in the form of intervals (denoted in blue).}
    \label{fig:highLevelVis}
\end{figure}
In order to extend mixed-monotone systems theory to stochastic systems, one needs a stochastic notion of partial order.
Stochastic order addresses this requirement by establishing an order between probability measures~\cite{stocOrderShakedShantikumar}.
This extension of mixed-monotone decomposition to stochastic systems allows us to compute concentration inequalities that ensure that a distribution passed through a monotone mapping will reside within an interval with some probability. 

We structure the rest of this paper as follows:
Section~\ref{sec:prelimProbForm} introduces the preliminaries necessary to establish mixed-monotonicity for stochastic systems and the problems we wish to solve. 
Section~\ref{sec:mixedMonoStoc} extends mixed-monotone systems theory via stochastic order. 
We first address the one step overapproximation of the forward reach set by exploiting the properties of monotone functions and prove how to obtain its probability of residing within the overapproximation in Section~\ref{sec:oneStep} then extend the result to a finite time horizon in Section~\ref{sec:nStep}. 
Examples follow in Section~\ref{sec:Examples}. 


\section{Preliminaries and problem formulation}
\label{sec:prelimProbForm}

\subsection{Probabilistic Forward Reachability}

We denote vectors with lowercase variables, $v\in\mathbb{R}^n$, and matrices with uppercase variables, $V\in\mathbb{R}^{n\times m}$.
The element of a vector $v$ is $v_i$ and an element of a matrix $V$ is $V_{ij}$, where $i,j\in\mathbb{N}$. 
We restrict our analysis to element-wise partial orders, $\leq$, with respect to the positive orthant, $\R^{n_\circ}_+$.
Random vectors are in bold, $\mathbf{v}$. 
To define the probabilistic elements of the problem, we define a probability space $(\Omega, \mathcal{F}, \mathbb{P})$. 
The dynamical system is a function, $f: \R^{n_x}\times \R^{n_w}\rightarrow \R^{n_x}$, 
\begin{equation}
	\label{eq:nonlinearSys}
	\mathbf{x}_{k+1} = f(\mathbf{x}_{k},\mathbf{w}_k).
\end{equation} 
For a random initial condition $\mathbf{x}_0:\Omega\to\R^{n_x}$ and a random disturbance sequence, $\mathbf{w}_k:\Omega\to\R^{n_w}$, for $k\in \N_+$ we define the random vector $\mathbf{x}_k:\Omega\to\R^{n_x}$ in terms of $\mathbf{x}_0$ and $\mathbf{w}_k$ according to the state evolution equation~\eqref{eq:nonlinearSys}. 
We can derive from a random $\mathbf{h}:\Omega\to\R^n$ its probability measure $P_\mathbf{h} = \mathbb{P} \circ \mathbf{h}^{-1}$, where $\circ$ denotes a composition of two functions. 
We interpret its probability measure as $P_{\mathbf{h}}(A)=\mathbb{P}(\{\omega: \mathbf{h(\omega)}\in A\})$ for $A\subseteq\R^n$ as the probability that $\mathbf{h}$ obtains a value in $A$. 
We denote two random vectors, $\mathbf{x}$ and $\mathbf{y}$, being equal in distribution by $\mathbf{x}=_{\mathrm{st}}\mathbf{y}$. 

 Our goal, formalized in Problem~\ref{prob:multi-step-reachability}, is to recursively construct a sequence of sets $X_k\subseteq\R^{n_x}$ that capture the general evolution of $\mathbf{x}_k$ by obtaining a prescribed probability mass. 

\begin{prob}[Finite-step stochastic monotone forward reachability]
\label{prob:multi-step-reachability}
Suppose we have a discrete-time dynamical system as in~\eqref{eq:nonlinearSys},
a random initial state $\mathbf{x}_0:\Omega\to\R^{n_x}$, a sequence $\mathbf{w}_k:\Omega\to\R^{n_w}$ of random disturbances where $k = 0,\dotsc,N$.
The state and disturbance reside in sets $X_0 \times W_1 \times \cdots \times W_k$ respectively with probability $\ge 1-\delta$ under their joint measure for a known $\delta>0$.
We wish to determine
a sequence of sets $X_k$ and a sequence of probability values $\delta_k > 0$ that satisfy the sequence of concentration inequalities, $P_{\mathbf{x}_k}\left( X_k\right) \ge 1 - \delta_k$.
\end{prob}
The most efficient way to apply monotonicity methods to Problem~\ref{prob:multi-step-reachability} is to first solve the single step case, i.e. take $k=1$, and to proceed recursively to multi-step reachability by upper-bounding the original problem via Boole's inequality. 
While the single-step reachability is evidently a special case of Problem~\ref{prob:multi-step-reachability}, its significance to the analysis in this paper warrants a separate statement.
\begin{prob}[one-step stochastic monotone forward reachability]
\label{prob:one-step-reachability}
Suppose we have a discrete-time dynamical system as in~\eqref{eq:nonlinearSys},
a random initial state $\mathbf{x}_0:\Omega\to\R^{n_x}$, and a disturbance $\mathbf{w}:\Omega\to\R^{n_w}$. 
The state and disturbance reside in sets $X_0 \times W$ respectively with probability $\ge 1-\delta_0$ under the joint measure of $\textbf{x}_0$ and $\textbf{w}$ for a known $\delta_0>0$.
We wish to determine a set $X_1$ and probability, $\delta_1$, satisfying the concentration inequality,
\begin{equation}
    \label{eq:onestep_concentration}
    P_{\mathbf{x}_1}(X_1) \ge 1 - \delta_1.
\end{equation}
\end{prob}

\subsection{Mixed Monotone Systems}
Though more tractable than Problem~\ref{prob:multi-step-reachability}, Problem~\ref{prob:one-step-reachability} is still difficult to solve for geometrically complicated $X_k$ and for nonlinear systems like~\eqref{eq:nonlinearSys}.
Thus the literature opts for geometrically simple approximations that exploit certain system properties; in our case, we use multidimensional interval geometry and the property of mixed-monotonicity.

\begin{defn}[Mixed-Monotonicity in Discrete Time {\cite{coogan2020mixed,meyer2021interval}}]
\label{def:mixedMonotone}
Suppose we have a system, \eqref{eq:nonlinearSys}, and a monotone decomposition $g: \R^{n_x}\times \R^{n_w} \times \R^{n_x}\times \R^{n_w}\times \R^{n_w} \rightarrow \R^{n_x}$. 
The system is mixed-monotone with respect to $g$ if the following conditions hold:
\begin{enumerate}
    \item $f(x,w) = g(x,w,x,w)$, $\forall x\in\R^{n_x}$ and $w\in\R^{n_w}$. 
    \item $g(x,w,\hat x, \hat w)\leq g(y,v,\hat x, \hat w)$, $\forall x,y,\hat x\in\R^{n_x}$ such that $x \leq y$, and $w,v,\hat w\in\R^{n_w}$ such that $w \leq v$. 
    \item $g(x, w, \hat{y}, \hat{v}) \leq g(x, w, \hat{x}, \hat{w})$ $\forall x, \hat{x}, \hat{y} \in \mathcal{X}$ such that $\hat{x} \leq \hat{y}$, and $\forall w, \hat{w}, \hat{v} \in \R^{n_w}$ such that $\hat{w} \leq \hat{v}$. 
\end{enumerate}
\end{defn}
With this system description in hand, we can over-approximate forward reach sets using evolutions of the mixed-monotone system.  
\begin{lem}[Discrete time forward interval reach sets {\cite[Theorem 1]{coogan2015efficient}}]
    Suppose we have a system \eqref{eq:nonlinearSys} with decomposition function, $g$. 
    If $[\underline{x},\overline{x}]\subseteq \R^{n_x}$ and $[\underline{w},\overline{w}]\subseteq \R^{n_w}$ where $\underline{x},\overline{x}\in\R^{n_x}$ and $\underline{w},\overline{w}\in\R^{n_w}$, then
    \begin{equation}
        \label{eq:intervalReach}
         g(\underline{x},\underline{w},\overline{x}, \overline{w}) \leq f(x,w) \leq g(\overline{x},\overline{w},\underline{x},\underline{w}),
    \end{equation}
    $\forall x\in[\underline{x},\overline{x}]$ and $\forall w\in[\underline{w},\overline{w}]$. 
\end{lem}

The above over-approximation provides $X_k$ as hyperrectangles, i.e. $X_k = [\underline{x}_k,\overline{x}_k]$, recursively for $k = 0,\dotsc,N$. It is precisely this result for which we seek a stochastic analogue; with such a result in hand, we will be able to construct interval approximations $X_k$ of the stochastic reachable sets that solve Problem~\ref{prob:multi-step-reachability}.


\subsection{Stochastic Orders}

In order to use mixed-monotonicity methods to solve the stochastic reachability problem posed in Problems~\ref{prob:multi-step-reachability} and~\ref{prob:one-step-reachability}, we must lift the notion of mixed monotinicity from the deterministic case to the stochastic case. Our approach, in a nutshell, is to use a notion of \emph{stochastic order} instead of the componentwise vector partial orders used in the ordinary setting of monotone functions. We briefly introduce the stochastic ordering we use, namely that of~\cite{stocOrderShakedShantikumar}. 

It is reasonable to assert that a scalar random variable $\mathbf{x}$ is smaller than another scalar random variable $\mathbf{y}$ under the condition that $P_\mathbf{x}([\alpha,\infty))$ is dominated by $P_\mathbf{y}([\alpha,\infty))$ for all $\alpha$; informally, what this asserts is that $\mathbf{y}$ is always more likely to exceed a given value $\alpha$ than $\mathbf{x}$ is.
The stochastic order theory of~\cite{stocOrderShakedShantikumar} is basically a vectorial extension of this concept. To formally define this partial order, we require a vectorial extension of the sets $[\alpha,\infty)$ called \emph{upper sets}.
\begin{defn}[Upper set~{\cite[Definition II.1.2]{dolecki2016convergence}}]
A set, $U\subseteq \R^n$, is an upper set in $\R^n$ such that $\forall x\in U$ and $\forall y \in \R^n$, if $x\leq y$ and $x\in U$, then $y\in U$. 
\end{defn}
Notice that we have made a choice of vector partial order $\le$ in defining upper sets; in the sequel, whenever we mix deterministic partial orders and stochastic orders, the deterministic partial order should be understood to be the same order used to define the upper sets.
\begin{defn}[Stochastic Order of random vectors~{\cite[Theorem 6.B.1]{stocOrderShakedShantikumar}}]
    A random vector $\mathbf{x}: \Omega \rightarrow \mathbb R^{n}$ is smaller $\mathbf{y}: \Omega \rightarrow \mathbb R^{n}$ in stochastic order, denoted as $\mathbf{x}\leq_{\mathrm{st}} \mathbf{y}$, considering the following holds, 
    \begin{equation}
        \label{eq:stochastic_order_definition}
        P_\textbf{x}(U) \le P_\textbf{y}(U)
    \end{equation}
    for all upper sets $U\subseteq\R^n$. 
\end{defn}

Stochastic orders enjoy a theory of monotonicity very similar to that of deterministic partial orders, but with an additional subtlety owing to the fact that there are really two orders at play: the stochastic order, and the underlying deterministic order used to define the upper sets. A core result in stochastic order theory is that a stochastic order between random variables is preserved by functions monotone with respect to the underlying deterministic order.
\begin{lem}[Stochastic order through a monotone function~{\cite[Theorem 6.B.16.b]{stocOrderShakedShantikumar}}]
    \label{lem:MonotoneStochOrder}
    Let $\mathbf{x}_i: \Omega \rightarrow \R^{n_{i}}$ and $\mathbf{y}_i: \Omega \rightarrow \R^{n_{i}}$ for $i = \{1,\hdots, m\}$ be a set of independent random vectors. 
    If $\mathbf{x}_i\leq_{\mathrm{st}}\mathbf{y}_i$ for $i=\{1,\hdots,m\}$, for any increasing function, $g: \R^{n_1\times\cdots\times n_m}\rightarrow\R$, we have, 
    \begin{equation}
        g(\mathbf{x}_1,\hdots,\mathbf{x}_m) \leq_{\mathrm{st}} g(\mathbf{y}_1,\hdots,\mathbf{y}_m)
    \end{equation}
\end{lem}
We can provide an equivalent result for a decreasing function as follows. 
\begin{corr}
\label{corr:MonotoneStochOrderRev}
    Let $\mathbf{x}_i: \Omega \rightarrow \R^{n_{i}}$ and $\mathbf{y}_i: \Omega \rightarrow \R^{n_{i}}$ for $i = \{1,\hdots, m\}$ be a set of independent random vectors. 
    If $\mathbf{x}_i\leq_{\mathrm{st}}\mathbf{y}_i$ for $i=\{1,\hdots,m\}$, for any decreasing function, $g: \R^{n_1\times\cdots\times n_m}\rightarrow\R$, we have, 
    \begin{equation}
        g(\mathbf{x}_1,\hdots,\mathbf{x}_m) \geq_{\mathrm{st}} g(\mathbf{y}_1,\hdots,\mathbf{y}_m)
    \end{equation}
\end{corr}
\begin{proof}
    This follows similarly to the poof of Lemma~\ref{lem:MonotoneStochOrder}~\cite[Theorem 6.B.16.b]{stocOrderShakedShantikumar}.
\end{proof}

In the sequel, we will need to establish a stochastic order between random vectors by establishing individual orders on their components. In order two do so, it is necessary that the underlying order $\le$ be with respect to an orthant cone (which we will assume from now on), and that the random vectors be expressed according to the following conditional form, called the \emph{standard construction} of the vector.
\begin{defn}[Standard construction of a random vector~{\cite[Sec. 6.B.3]{stocOrderShakedShantikumar}}]
    \label{def:stdConstruction}
    Suppose we have a random vector $\mathbf{x}=(\mathbf{x}_1,\dotsc,\mathbf{x}_n)$ with probability measure $\mathbb{P}_{\mathbf{x}}$. 
    If $u_1,\hdots,u_n$ are independently sampled uniform random variables between 0 and 1, i.e. $\mathrm{U}([0,1])$ where

    \begin{subequations}
        \begin{align}
            \Tilde{x}_1 
            &= 
            \inf\{x_1: 
            P_{\mathbf{x}_1}(
            x_1)\geq u_1
            \},\\ 
            \Tilde{x}_i 
            &= 
            \inf\{x_i: 
            P_{\mathbf{x}_i| \mathbf{x}_1 = \Tilde{x}_1,\hdots,\mathbf{x}_{i-1}=\Tilde{x}_{i-1}}(x_i)
            \geq u_i\},
        \end{align}
    \end{subequations}
    for $i=\{1,\hdots,n\}$ result in the samples of $\Tilde{\mathbf{x}}$, then we say $\Tilde{x}$ is the standard construction of $\mathbf{x}$ where $\Tilde{\mathbf{x}}=_{\mathrm{st}}\mathbf{x}$. 
\end{defn}
\begin{lem}[Natural construction of stochastic order]
\label{lem:naturalConstruction}
Suppose we have two random vectors $\mathbf{x}$ and $\mathbf{y}$ with respective standard constructions $\Tilde{x}$ and $\Tilde{y}$ via Definition~\ref{def:stdConstruction}. 
If  $\Tilde{x}_i \leq \Tilde{y}_i, \forall i=\{1,\hdots,n\}$ for every sample, then $\mathbf{x}\leq_{\mathrm{st}}\mathbf{y}$. 
\end{lem}

\section{Towards a Stochastic-order Theory of Monotone Dynamics}
\label{sec:mixedMonoStoc}

With the meaning of order between random vectors-- that is, the stochastic order-- established, we now turn to the task of extending mixed-monotonicity to the stochastic order case. Recall that the two of the three defining qualities of a mixed-monotone decomposition function are expressed with respect to a deterministic vector partial order $\le$; a sensible \emph{ansatz} for a stochastic extension is to replace this deterministic partial order with the stochastic order that it underlies. With this notion, we propose the following definition for a stochastic mixed-monotone system.

\begin{defn}
\label{def:mixedMonotoneStochastic}
Given a continuous function, $g: \R^{n_x}\times \R^{n_w} \times \R^{n_x}\times \R^{n_w}\times \R^{n_w} \rightarrow \R^{n_x}$, a stochastic system, \eqref{eq:nonlinearSys}, is mixed-monotone with respect to $g$ via the usual stochastic order, provided it satisfies the following conditions,
\begin{enumerate}
    \item Given independent random vectors $\mathbf{x}$ and $\mathbf{w}$, the nonlinear system, $f$, and the decomposition $g$, then $f$ and $g$ are equal in distribution, $f(\mathbf{x},\mathbf{w}) =_{\mathrm{st}} g(\mathbf{x},\mathbf{w},\mathbf{x},\mathbf{w})$.
    \item Suppose we have independent random vectors $\mathbf{x}$, $\hat{\mathbf{x}}$, $\mathbf{w}$,  $\hat{\mathbf{w}}$, $\mathbf{y}$, $\mathbf{v}$. 
    Given $\mathbf{x}\leq_{\mathrm{st}} \mathbf{y}$, and $\mathbf{w}\leq_{\mathrm{st}} \mathbf{v}$, results in $g(\mathbf{x},\mathbf{w},\hat{\mathbf{x}}, \hat{\mathbf{w}})\leq_{\mathrm{st}} g(\mathbf{y},\mathbf{v},\hat{\mathbf{x}}, \hat{\mathbf{w}})$.
    \item Suppose we have independent random vectors $\mathbf{x}$, $\hat{\mathbf{x}}$, $\mathbf{w}$,  $\hat{\mathbf{w}}$, $\hat{\mathbf{y}}$, $\hat{\mathbf{v}}$. 
    Given $\hat{\mathbf{x}}\leq_{\mathrm{st}} \hat{\mathbf{y}}$, and $\hat{\mathbf{w}}\leq_{\mathrm{st}} \hat{\mathbf{v}}$, results in $g(\mathbf{x},\mathbf{w},\hat{\mathbf{y}}, \hat{\mathbf{v}})\leq_{\mathrm{st}} g(\mathbf{x},\mathbf{w},\hat{\mathbf{x}}, \hat{\mathbf{w}})$.
\end{enumerate}
\end{defn}

With a definition in hand, we are compelled to consider: do any stochastic mixed-monotone systems, as we have defined them, exist? Indeed, they abound. Here is one way to construct one. Take a deterministic mixed-monotone system, and replace the initial state vector with a random vector, and likewise for the disturbances. The resulting stochastic system is a stochastic mixed-monotone system. The following results demonstrate this fact.

\begin{prop}[Condition 1 in Definition~\ref{def:mixedMonotone}]
Given independent random vectors $\mathbf{x}$ and $\mathbf{w}$ with respective standard constructions $\Tilde{x}$ and $\Tilde{w}$, the nonlinear system, $f$, and its decomposition function, $g$. 
If $f_i(\Tilde{x},\Tilde{w}) = g_i(\Tilde{x},\Tilde{w},\Tilde{x},\Tilde{w})$, then $f(\mathbf{x},\mathbf{w}) =_{\mathrm{st}} g(\mathbf{x},\mathbf{w},\mathbf{x},\mathbf{w})$.
    
\end{prop}
\begin{proof}
    This follows directly by the standard construction of a random vector in Definition~\ref{def:stdConstruction} and the construction of the deterministic variant of the decomposition in Definition~\ref{def:mixedMonotone}.
\end{proof}

\begin{prop}[Condition 2 in Definition~\ref{def:mixedMonotone}]
    \label{prop:cond2MixedMonotone}
    Suppose we have independent random vectors $\mathbf{x}$, $\hat{\mathbf{x}}$, $\mathbf{w}$,  $\hat{\mathbf{w}}$, $\mathbf{y}$, $\mathbf{v}$ with their respective standard constructions, the nonlinear system, $f$, and its decomposition function, $g$. 
    If $g(\Tilde{x},\Tilde{w},\hat{\Tilde{x}}, \hat{\Tilde{w}})\leq g(\Tilde{y},\Tilde{v},\hat{\Tilde{x}}, \hat{\Tilde{w}})$ where $\Tilde{x}\leq \Tilde{y}$ and $\Tilde{w}\leq \Tilde{v}$, then $g(\mathbf{x},\mathbf{w},\hat{\mathbf{x}}, \hat{\mathbf{w}})\leq_{\mathrm{st}} g(\mathbf{y},\mathbf{v},\hat{\mathbf{x}}, \hat{\mathbf{w}})$ where $\mathbf{x}\leq_{\mathrm{st}} \mathbf{y}$ and $\mathbf{w}\leq_{\mathrm{st}} \mathbf{v}$. 
\end{prop}
\begin{proof}
    Provided that we have standard constructions of all the random vectors involved via Definition~\ref{def:stdConstruction}, we know that $\Tilde{x}\leq \Tilde{y}$ and $\Tilde{w}\leq \Tilde{v}$ satisfy Lemma~\ref{lem:naturalConstruction}, then $\mathbf{x}\leq_{\mathrm{st}}\mathbf{y}$ and $\mathbf{w}\leq_{\mathrm{st}} \mathbf{v}$. 
    Given the stochastic orders by natural constructions and we assume them being independent, by Lemma~\ref{lem:MonotoneStochOrder}, $g(\mathbf{x},\mathbf{w},\hat{\mathbf{x}}, \hat{\mathbf{w}})\leq_{\mathrm{st}} g(\mathbf{y},\mathbf{v},\hat{\mathbf{x}}, \hat{\mathbf{w}})$. 
\end{proof}

\begin{corr}[Condition 3 in Definition~\ref{def:mixedMonotone}]
    Suppose we have independent random vectors $\mathbf{x}$, $\hat{\mathbf{x}}$, $\mathbf{w}$,  $\hat{\mathbf{w}}$, $\hat{\mathbf{y}}$, $\hat{\mathbf{v}}$ with their respective standard constructions, the nonlinear system, $f$, and its decomposition function, $g$. 
    If $g(\Tilde{x},\Tilde{w},\hat{\Tilde{y}}, \hat{\Tilde{v}})\leq g(\Tilde{x},\Tilde{w},\hat{\Tilde{x}}, \hat{\Tilde{w}})$ where $\hat{\Tilde{x}}\leq \hat{\Tilde{y}}$ and $\hat{\Tilde{w}}\leq \hat{\Tilde{v}}$, then $g(\mathbf{x},\mathbf{w},\hat{\mathbf{y}}, \hat{\mathbf{v}})\leq_{\mathrm{st}} g(\mathbf{x},\mathbf{w},\hat{\mathbf{x}}, \hat{\mathbf{w}})$ where $\hat{\mathbf{x}}\leq_{\mathrm{st}} \hat{\mathbf{y}}$ and $\hat{\mathbf{w}} \leq_{\mathrm{st}} \hat{\mathbf{v}}$.
\end{corr}

\begin{proof}
    Arguments follow the proof in Theorem~\ref{prop:cond2MixedMonotone}.
\end{proof}
 
Having established the existence of stochastic mixed-monotone systems, we turn to the order preservation properties of their trajectories.
First, note that the random vectors $\hat{\mathbf{x}}$, $\hat{\mathbf{w}}$, $\mathbf{y}$, and $\mathbf{v}$ and their relationship to $\mathbf{x}$ and $\mathbf{w}$ are absent in Definition~\ref{def:mixedMonotoneStochastic}, which is crucial for an initializing an execution. 
The following result provides a way prove that two random vectors are equal.
\begin{lem}[Construction of $\hat{\mathbf{x}}$ via $\mathbf{x}$]
If we have random vectors $\mathbf{x}$ and $\hat{\mathbf{x}}$ such that $\mathbf{x}\leq_{\mathrm{st}}\hat{\mathbf{x}}$ and $\mathbb{E}[h_i(\mathbf{x}_i)] = \mathbb{E}[h_i(\mathbf{y}_i)]$ for strictly increasing functions, $h_i: \R \rightarrow \R,\ i=\{1,\hdots,n\}$, then $\mathbf{x} = _{\mathrm{st}} \hat{\mathbf{x}}$. 
\end{lem}
\begin{proof}
    Follows directly from~\cite[Theorem 6.B.19]{stocOrderShakedShantikumar}. 
\end{proof}
Thus, having an initial $\mathbf{x}_0$ is sufficient to construct an initial $\hat{\mathbf{x}}_0$. 
A similar argument results for the disturbances, where having intermediate random vectors $\mathbf{w}_k$ is sufficient for constructing $\hat{\mathbf{w}}_k$ for $k=\{0,\hdots,N-1\}$.
As Definition~\ref{def:mixedMonotoneStochastic} alludes, an execution of the mixed-monotone system for a finite horizon $N$ should preserve the monotonicity via the usual stochastic order.
We now prove the execution of a stochastic, mixed monotone system to recover interval reach sets, we first define southeast order and define a stochastic order variant of it.
\begin{defn}[Southeast order~\cite{coogan2020mixed}]
    For $(x,\hat x),(y,\hat y)\in\R^{2n}$, $(x,\hat x)\leq_{\mathrm{SE}}(y,\hat y)$ means $x\leq y$ and $\hat x\geq\hat y$.
\end{defn}
\begin{defn}
        For $(\mathbf{x},\mathbf{\hat x}),(\mathbf{y},\mathbf{\hat y})\in\R^{2n}$, $(\mathbf{x},\mathbf{\hat x})\leq_{\mathrm{stSE}}(\mathbf{y},\mathbf{\hat y})$ means $\mathbf{x}\leq_{\mathrm{st}}\mathbf{y}$ and $\mathbf{\hat x}\geq_{\mathrm{st}}\mathbf{\hat y}$.
\end{defn}
\begin{thm}[Execution of a stochastic mixed-monotone system]
\label{thm:embeddingExecution}
Suppose we have a mixed monotone system of the form, 
\begin{equation}
    \label{eq:monotoneSysDesc}
    \begin{bmatrix}
        \mathbf{x}_{k+1}\\
        \hat{\mathbf{x}}_{k+1}
    \end{bmatrix} = 
    \begin{bmatrix}
    g(\mathbf{x}_k,\mathbf{w}_k,\hat{\mathbf{x}}_k,\hat{\mathbf{w}}_k)\\ 
    g(\hat{\mathbf{x}}_k,\hat{\mathbf{w}}_k,\mathbf{x}_k,\mathbf{w}_k)
    \end{bmatrix}.
\end{equation}
If $(\mathbf{x},\hat{\mathbf{x}})\leq_{\mathrm{stSE}}(\mathbf{x}',\hat{\mathbf{x}}')$ and  $(\mathbf{w},\hat{\mathbf{w}})\leq_{\mathrm{stSE}} (\mathbf{w}',\hat{\mathbf{w}}')$, then $(\mathbf{x}_{k+1},\hat{\mathbf{x}}_{k+1})\leq_{\mathrm{stSE}}(\mathbf{x}'_{k+1},\hat{\mathbf{x}}'_{k+1})$. That is, 
\begin{subequations}
    \begin{align}
        g(\mathbf{x},\mathbf{w},\hat{\mathbf{x}}',\hat{\mathbf{w}}') \leq_{\mathrm{st}} g(\mathbf{x}',\mathbf{w}',\hat{\mathbf{x}},\hat{\mathbf{w}}),\label{eq:monotoneSysUp}\\
        g(\hat{\mathbf{x}},\hat{\mathbf{w}},\mathbf{x}',\mathbf{w}') \geq_{\mathrm{st}} g(\hat{\mathbf{x}}',\hat{\mathbf{w}}',\mathbf{x},\mathbf{w}),\label{eq:monotoneSysDown}
    \end{align}
\end{subequations}
respectively.
\end{thm}
\begin{proof}
    By Definition~\ref{def:mixedMonotoneStochastic}, Condition 2, we have 
    \begin{equation}
        g(\mathbf{x},\mathbf{w},\hat{\mathbf{y}},\hat{\mathbf{v}}) \leq_{\mathrm{st}} g(\mathbf{x}',\mathbf{w}',\hat{\mathbf{y}},\hat{\mathbf{v}}),
    \end{equation}
    for $\hat{\mathbf{y}}$ and $\hat{\mathbf{v}}$. 
    Likewise, by Definition~\ref{def:mixedMonotoneStochastic}, Condition 3, we have 
    \begin{equation}
        g(\mathbf{y},\mathbf{v},\hat{\mathbf{x}}',\hat{\mathbf{w}}') \leq_{\mathrm{st}} g(\mathbf{y},\mathbf{v},\hat{\mathbf{x}},\hat{\mathbf{w}})
    \end{equation}
    for $\mathbf{y}$ and $\mathbf{v}$.
    Thus, by Lemma~\ref{lem:MonotoneStochOrder} where $\hat{\mathbf{y}}\geq_{\mathrm{st}}\hat{\mathbf{x}}$ and $\hat{\mathbf{v}}\geq_{\mathrm{st}}\hat{\mathbf{w}}$, thus  $\hat{\mathbf{y}}\geq_{\mathrm{st}}\hat{\mathbf{x}}'$ and $\hat{\mathbf{v}}\geq_{\mathrm{st}}\hat{\mathbf{w}}'$, we have \eqref{eq:monotoneSysUp}. 
    The argument to prove \eqref{eq:monotoneSysDown} follows similarly but uses Corollary~\ref{corr:MonotoneStochOrderRev}. 
\end{proof}

To summarize, the decomposed dynamics are stochastically monotone in the southeast stochastic order. To turn this fact into a computable algorithm that can solve Problem 1, we require an additional result to derive concentration inequalities from monotone functions. This is derived in the next section.

\section{From Stochastic Order to Concentration Inequality}
\label{sec:oneStep}
Having established the foundations of mixed monotonicity with respect to stochastic orders, we use the stochastic order property to establish the concentration inequalities laid out in~\eqref{eq:onestep_concentration} and thereby solve Problem~\ref{prob:one-step-reachability}. The central property we use is that stochastic orders imply a preservation of probability mass for upper sets. Indeed, if $X_k$ in Problem~\ref{prob:one-step-reachability} were allowed to be an upper set, then establishing a stochastic order between $\textbf{x}$ and $f(\textbf{x})$ would immediately solve the problem by the property~\eqref{eq:stochastic_order_definition}; however, since upper sets are by definition unbounded, they are usually inappropriate for the purposes of safety verification. To arrive at a compact estimate for the stochastic reach set, we use an intersection of an upper set and a lower set and use a stochastic order on each part to establish a bound on the probability mass of the intersection.

For our upper and lower sets, we use orthant cones in $\R^{n}$ with respect to the partial orders $\le$ and $\ge$, yielding intersections with interval geometry. For this, we require the following elementary fact about approximation of quantiles for the output of a monotone function.

\begin{lem}
\label{lem:monotone-cone-preservation}
Let $\textbf{x}:\Omega\to\R^n$ be a random variable, $f:\R^n\to\R^m$ a function monotone with respect to the partial order $\le$, and $z\in\R^n$. Then 
$P(f(x) \le f(z)) \ge P(x \le z)$.
\end{lem}
\begin{proof}
\begin{subequations}
\begin{align}
    & \qquad \mathbb{P}(\{\omega : f(\textbf{x}(\omega)) \le f(z)\})\\
    &=
    \int \ind{f(\textbf{x}(\omega)) \le f(z)}  
    d\mathbb{P}(\omega)
    \label{eq:vecmono1}\\
    &= 
    \int
    \ind{f(\textbf{x}(\omega)) \le f(z)
    \text{ and } 
    \textbf{x}(\omega) \le z
    }\nonumber\\
    &+
    \ind{f(\textbf{x}(\omega)) \le f(z)
    \text{ and } 
    \textbf{x}(\omega) > z
    }
    d\mathbb{P}(\omega)
    \label{eq:vecmono2}\\
    &= 
    \int
    \ind{x(\omega) \le z}\nonumber\\
    &+
    \ind{f(\textbf{x}(\omega)) \le f(z)
    \text{ and } 
    \textbf{x}(\omega) > z
    }
    d\mathbb{P}(\omega)
    \label{eq:vecmono3}\\
    &\ge
    \int
    \ind{\textbf{x}(\omega) \le z}
    d\mathbb{P}(\omega)
    \label{eq:vecmono4}\\
    &= \mathbb{P}(\{\omega : \textbf{x}(\omega) \le z\}),\label{eq:vecmono5}
\end{align}
\end{subequations}
We note that~\eqref{eq:vecmono3} follows from~\eqref{eq:vecmono2} by monotonicity and from~\eqref{eq:vecmono3} to~\eqref{eq:vecmono4} by the nonnegativity of the indicator.
\end{proof}

Now, suppose we have an interval $[\xlo, \xhi]$ such that $P_X([\xlo,\xhi])\ge 1-\epsilon$. 
By applying Lemma~\ref{lem:monotone-cone-preservation} two times, one for each of the two orthant cones whose intersection forms $[\xlo, \xhi]$, we obtain the following bound on the probability mass of $[f(\xlo), f(\xhi)]$.

\begin{lem}
\label{lem:one-minus-two-delta}
Let $\mathbf{x}:\Omega\to\R^n$ be a random variable, 
$[\xlo, \xhi]$ an interval such that 
$P_\textbf{x}([\xlo,\xhi]) \ge 1-\delta$, 
and $f:\R^n\to\R^m$ a function monotone with respect to the partial order $\le$. Then $P_{f(\textbf{x})}([f(\xlo),f(\xhi)]) \ge 1-2\delta$.
\end{lem}
\begin{proof}
Since $P_\textbf{x}([\xlo,\xhi])\ge 1-\epsilon$, it follows that 
$\mathbb{P}(\{\omega : \textbf{x}(\omega) \le \xhi\}) \ge 1-\epsilon$
and
$\mathbb{P}(\{\omega: \textbf{x}(\omega) \ge \xlo\}) \ge 1-\epsilon$. In turn, from Lemma~\ref{lem:monotone-cone-preservation}
we find that
$\mathbb{P}(\{\xi : f(\textbf{x}(\omega) \le f(\xhi)\}) \ge 1-\epsilon$
and
$\mathbb{P}(\{\omega : f(\textbf{x}(\omega)) \ge f(\xlo)\}) \ge 1-\epsilon$.
Then, from the additive measure property
$\mathbb{P}(A\cup B) = \mathbb{P}(A) + \mathbb{P}(B) - \mathbb{P}(A\cap B)$ and consequently
$\mathbb{P}(A\cap B) = \mathbb{P}(A) + \mathbb{P}(B) - \mathbb{P}(A\cup B)$,
it follows that the probability mass of $[f(\xlo),f(\xhi)]$ is
\begin{equation}
\begin{aligned}
    &\quad P_{f(\textbf{x})}([f(\xlo),f(\xhi)])\\
    &=
    \mathbb{P}(\{\omega : f(\textbf{x}(\omega)) \ge f(\xlo)\})
    + 
    \mathbb{P}(\{\omega : f(\textbf{x}(\xi)) \le f(\xhi)\})\\
    &-
    \mathbb{P}(\{\omega : f(\textbf{x}(\omega)) \ge f(\xlo)\text{ or } f(\textbf{x}(\omega)) \leq f(\xhi)\})\\
    &\ge (1-\epsilon) + (1-\epsilon) - 1 = 1-2\epsilon.
\end{aligned}
\end{equation}
\end{proof}
Lemma~\ref{lem:one-minus-two-delta} immediately suggests a strategy by which to solve Problem~\ref{prob:one-step-reachability} with interval geometry: apply $f$ to the upper and lower bounds of an interval containing $1-\delta$ probability mass of $\textbf{x}_0$, and the resulting interval solves Problem~\ref{prob:one-step-reachability} with $\delta_0 = \delta$, $\delta_1 = 2\delta$. The precise details of this approach are explored in the next section.


\section{Algorithms for computing forward reachable set of a desired probability via stochastic order}
\label{sec:nStep}

With formalism of stochastic order for mixed-monotone systems and deriving the propagation interval reach sets and the resuling probability of the state residing within it, we arrive at an algorithm which addresses Problem~\ref{prob:multi-step-reachability}.
\begin{algorithm}
\caption{Interval Sets for \eqref{eq:nonlinearSys} via Mixed Monotonicity}
\label{alg:prop}
\begin{algorithmic}[1]
\Require $P_{\mathbf{x}_0}$, $P_{\mathbf{w}}$, $\delta_{0}$, $\delta_{w}$, $N$, \eqref{eq:nonlinearSys}, \eqref{eq:monotoneSysDesc}
\Ensure $\{(x_k,[\underline{x}_k,\overline{x}_k],\delta_k)\}_{k=1}^{N}$.
\State $x_0\sim P_{\mathbf{x}_0}$ 
\State Find $[\underline{x}_0,\overline{x}_0]$ such that  $P_{\mathbf{x}_0}([\underline{x}_0,\overline{x}_0])\geq 1-\delta_{0}$, \label{algLine:intervalX0}
\State Find $[\underline{w},\overline{w}]$ such that  $P_{\mathbf{w}}([\underline{w},\overline{w}])\geq 1-\delta_{w}$,\label{algLine:intervalW}
\State $\{(x_k,[\underline{x}_k,\overline{x}_k],\delta_k)\}_{k=0}^{N}\leftarrow (x_0,[\underline{x}_0,\overline{x}_0],\delta_{0})$
\For{$k={0,\hdots,N-1}$}
\State $w_k\sim P_{\mathbf{w}}$
\State Evaluate \eqref{eq:nonlinearSys},
\State $\delta_{x_{k+1}} = 1-(1-2\delta_{k})(1-\delta_{w})$ \label{algLine:xIntervalProb}
\State Evaluate \eqref{eq:intervalReach} with $[\underline{x}_k,\overline{x}_k]$ and $[\underline{w},\overline{w}]$. 
\State $[\underline{x}_{k+1},\overline{x}_{k+1}] =$
\Statex \hspace{7em}$[g(\underline{x}_k,\underline{w}_k,\overline{x}_k,\overline{w}_k),g(\overline{x}_k,\overline{w}_k,\underline{x}_k,\underline{w}_k)]$
\State $\{(x_k,[\underline{x}_k,\overline{x}_k],\delta_k)\}_{k=0}^{N}\hspace{-0.6em}\leftarrow (x_{k+1},[\underline{x}_{k+1},\overline{x}_{k+1}],\delta_{k+1})$
\EndFor
\end{algorithmic}
\end{algorithm}
At a high level, Algorithm~\ref{alg:prop} takes in the initial measure on the state and disturbance as well as probability the state and disturbances resides within some intervals. 
Doing so allows on to obtain the interval reach set and the probability of the state lying within it by simulating the system and its decomposition for a finite horizon. 

To initialize the algorithm, we presume that the elements of the initial state $\mathbf{x}$ and disturbance $\mathbf{w}$ are independent in addition to the vectors also being independent. 
This independence assumption allows us one to define the induced measure as the products of the elements of the random vector's induced measures. 
Consider $\mathbf{z} = [\mathbf{x}\ \mathbf{w}]^\intercal\in\R^{p}$ where $p = n+m$,
\begin{equation}
    P_{\mathbf{z}}(z) = P_{\mathbf{z}_{1}}(z_{1})\cdots P_{\mathbf{z}_{p}}(z_{p})\geq 1-\delta_{z}. 
\end{equation}
Therefore we obtain intervals for the elements of random variables by 
\begin{subequations}
    \begin{align}
        P_{\mathbf{z}_{1}}([\underline{z}_{1},\overline{z}_{1}]) = P_{\mathbf{z}_1}(\overline{z}_1) - P_{\mathbf{z}_1}(\underline{z}_1)\geq 1 - \delta_{z_1}, 
    \end{align}
\end{subequations}
where $\delta_{z_i} = \delta_{z}^{1/p}$.
Lines~\ref{algLine:intervalX0} and~\ref{algLine:intervalW} of Algorithm~\ref{alg:prop} conduct these steps such that the joint probability of residing within intervals is,
\begin{align}
    \label{eq:intervalProbAlg}
    P_{\mathbf{z}}([\underline{z},\overline{z}]) = P_{\mathbf{z}_{1}}([\underline{z}_{1},\overline{z}_{1}])\cdots P_{\mathbf{z}_{p}}([\underline{z}_{p},\overline{z}_{p}])\geq 1 - \delta_{z}.  
\end{align}

One can iterate the one step lower bound of the probability of residing within an interval via Lemma~\ref{lem:one-minus-two-delta} over the time horizon as shown in line~\ref{algLine:xIntervalProb} of Algorithm~\ref{alg:prop}. 
However note that the probability continually decreases as the time horizon progresses, and perhaps even becomes negative. 
We can restrict the $1-2\delta_k$ bound to the interval $[0,1]$ and holds irrespective to the underlying monotone nonlinearity. 
We leave it to future work to tighten this bound and incorporate properties of the underlying monotone mapping.
Thus Algorithm~\ref{alg:prop} provides the overapproximations of the forward reach sets in terms of intervals, $X_k$, with the probability $1-\delta_k$ that the state resides within the set.


\section{Examples}
\label{sec:Examples}

We present our approach on two aerospace problems of interest.
For both examples, we presume that we are given continuous system and its decomposition which is discretized with respect to time. 
Both systems also have a  disturbance input, which is then discretized with respect to time. 
All simulations were run in MATLAB 2023a on an Apple M1 Macbook Pro with 16GB of RAM. 

\subsection{CWH dynamics}
The first example we consider are the Clohessy-Wiltshire-Hill (CWH) equations \cite{wiesel1989spaceflight}. 
These equations describe the equations of a deputy satellite rendezvousing with a chief satellite in the same elliptical orbit.  
\begin{align}
    \ddot{x} - 3 \omega x - 2 \omega \dot{y} = m_{d}^{-1}F_{x},\qquad\ddot{y} + 2 \omega \dot{x} = m_{d}^{-1}F_{y}.
  \label{eq:2d-cwh}
\end{align}
We consider the chief at the origin, the position of the deputy is $x,y \in \mathbb{R}$. The orbital frequency is $\omega = \sqrt{\mu/R_{0}^{3}}$ which consists of $\mu$, the gravitational constant, and $R_{0}$, the orbital radius of the spacecraft.
See~\cite{lesser2013stochastic} for further details and numerical values.

We formulate this as a LTI system discretized with a zero order hold and  close the loop using a LQR controller with weights $Q = \mathrm{diag}(1, 1, 10, 10])$ and $R = \mathrm{diag}(10,10)$.
\begin{align}
    \mathbf{x}_{k+1} = \hat{A}\mathbf{x}_k + G\mathbf{w}_k,
\end{align}
where $\hat{A} = A - BK\in\R^{4\times 4}$ with system matrix, $A$, input matrix $B\in\R^{4\times 2}$, control gain matrix via LQR, $K\in\R^{2\times 4}$, and disturbance matrix, $G\in\R^{4\times 4}$. 
The state vector is $x = [x_1\ x_2\ \dot{x}_1\ \dot{x}_2] \in \mathbb{R}^{4}$ and the input is $u = [F_{x}\ F_{y}] \in \mathcal{U}\subset\mathbb{R}^{2}$.
We initialize the state from a Gaussian distribution with mean $\mu_{\mathbf{x}_0} = [10\ -5\ 0\ 0]^\intercal$ and covariance $\Sigma_{\mathbf{x}_0} = \mathrm{diag}(0.5, 0.5, 0.01, 0.01)$. 
We also perturb the system with Gaussian disturbance, $\mathbf{w}_k\in \mathbb{R}^{4}$, with mean $\overline{\mu}_{\mathbf{w}_k} = {[0\ 0\ 0\ 0]}^\intercal$ and covariance matrix $\Sigma_{\mathbf{w}_k} = \mathrm{diag}(0.1,0.1,0.01,0.01)$.

The decomposition of a discrete time linear system is provided by the function, 
\begin{align}
    d(z,w,\hat{z},\hat{w}) =& \max(\hat{A},0) z + \max(G,0) w\nonumber\\
    &+ \min(\hat{A},0) \hat{z} + \min(\hat{G},0)\hat{w}.
\end{align}
We apply Algorithm~\ref{alg:prop} for a horizon $N=5$ where the sampling time is $T_s = 20 s$.
Thus the time horizon is 100 seconds.
We presume that the initial state and the disturbance reside within intervals, \eqref{eq:intervalProbAlg}, with $95\%$ probability (i.e. $1-\delta_0 = 0.95$).

Figure~\ref{fig:Example1F1} shows the trajectory of both the system and the evolution of the system in blue and the overapproximation of the forward stochastic reach set. 
Note that with linear dynamics and an initial 95\% probability in the interval reach sets, we have a majority of the trajectories of the system, in blue, contained within the interval reach set, in red. 
However, the third subplot shows that the probability of staying within the interval reach set drops to zero after 80 seconds.
Thus, while the concentration result in Lemma~\ref{lem:one-minus-two-delta} provides us a guarantee, irrespective of the underlying monotone mapping, incorporating more of the underlying system structure and stochasticity should improve this lower bound. 

\begin{figure}[t!]
    \centering
    \includegraphics[width=\linewidth,trim={0.5cm 0.25cm 0.5cm 0.25cm}, clip]{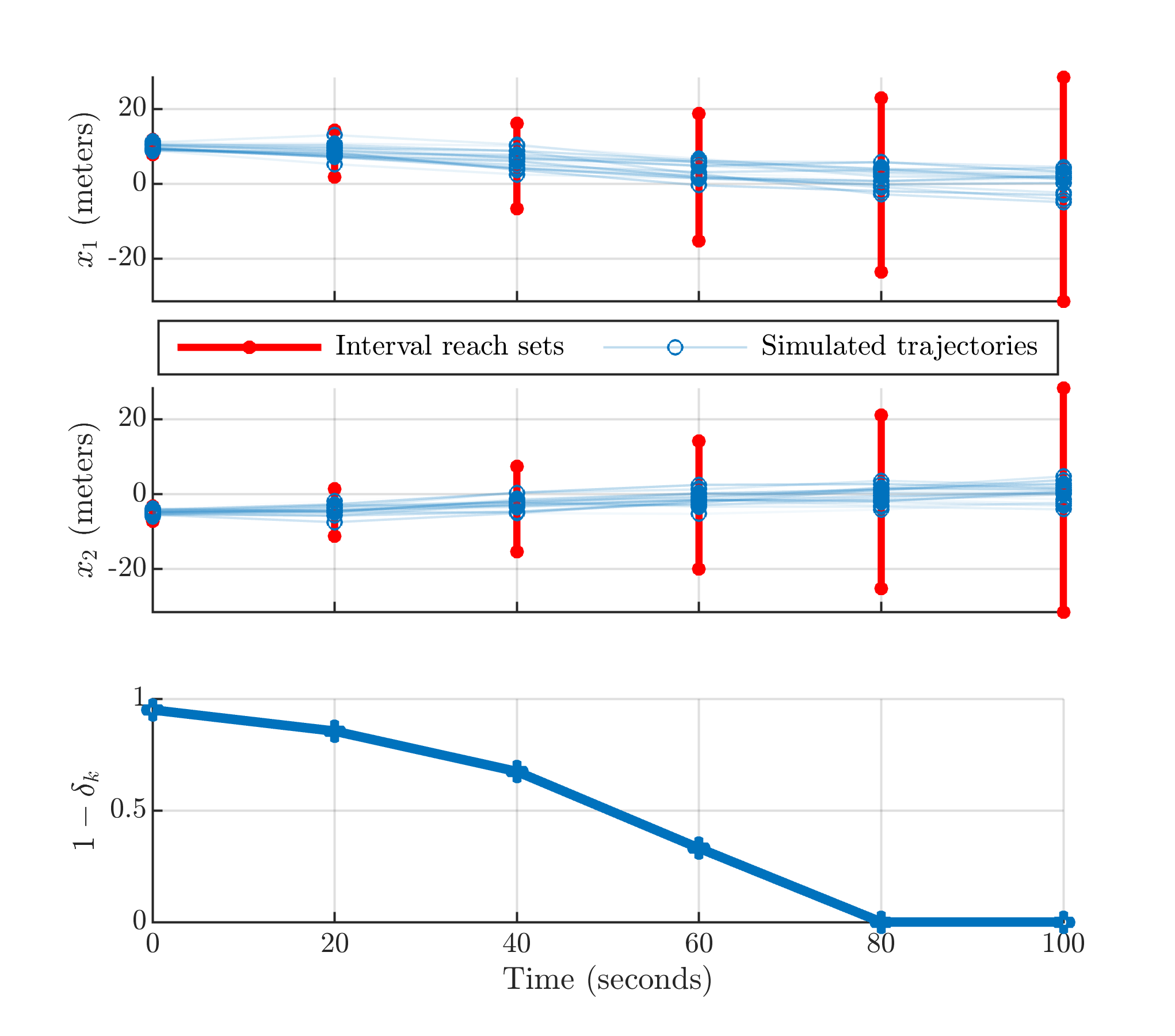}
    \caption{
    Resulting spacecraft trajectory subjected to random initial condition and random disturbances starting with $95\%$ probability (i.e. $1-\delta_0 = 0.95$) of the state and disturbance lying within their intervals, , \eqref{eq:intervalProbAlg}. 
    Note that while the theoretical bound we recursively update in line~\ref{algLine:xIntervalProb} of Algorithm~\ref{alg:prop} drops to zero past 80 seconds, a majority of the simulated trajectories still reside within the interval reach set.}
    \label{fig:Example1F1}
\end{figure}

\subsection{Line-of-sight of a 7 dimensional spacecraft}

We adapt a 7-dimensional spacecraft example as well as its mixed monotone decomposition from~\cite{abate2021decomposition}.
The system is composed of four elements, starting with the angular velocity dynamics, 
\begin{equation}
    \dot{\omega} = J^{-1}(-(\omega\times J\omega) + u + \mathbf{w})
\end{equation}
where $J\in\R^3$ is the inertia matrix, $\omega\in\R^3$ is the angular velocity, $u\in\R^3$ is the input, and $\mathbf{w}\in\R^3$ is the disturbance.
The quaternion dynamics describe the evolution of the orientation, i.e. attitude, of the system, 
\begin{equation}
    \dot{q} = 
\frac{1}{2}
\begin{bmatrix}
-q_1 & -q_2 & -q_3 \\
q_0 & -q_3 & q_2 \\
q_3 & q_0 & -q_1 \\
-q_2 & q_1 & q_0
\end{bmatrix} \omega,
\end{equation}
where $q\in\R^4$ denotes the orientation of the satellite via the quaternion. 
Thus the state is $x = [q\ \omega]^\intercal\in\R^7$. 
We close the loop with a PD controller, 
\begin{equation}
    u_{pd} = \omega\times J\omega -k_p J
    \begin{bmatrix}
        2(q_2q_3 + q_0q_1)\\
        2(q_0q_2 - q_1 q_3\\
        0
    \end{bmatrix}
    - k_d J\omega,
\end{equation}
for $k_p,k_d>0$. 
This controller passes through a saturation function, 
\begin{equation}
    u = \frac{1}{2}\tanh(2u_{pd}). 
\end{equation}
We project the forward reachable set of a particular probability on the rotation about the z axis where a line-of-sight sensor is on the spacecraft. 
This is merely the conversion from the quaternion to satellite's orientation via Euler angles, which is taken to be 
\begin{equation}
    \label{eq:angle_sensor}
    \theta(q) = \arccos(1-2q_1^2-2q_2^2), 
\end{equation}
to recover a mixed monotone decomposition function of the system. 

In this implementation, we presume the inertia matrix, 
\begin{equation}
    J = 
    \begin{bmatrix}
        17.5 & -0.8 & 0.3\\
        -0.8 & 14.9 & 0.4\\
        0.3  & 0.4  & 20.8 
    \end{bmatrix},   
\end{equation}
with controller gains $k_p = 0.6$ and $k_d = 2.25$. 
We initialize the system with a Gaussian distribution with mean $\mu_{\mathbf{x}_0} = [\sqrt{3}/2\ 1/2\ 0\ 0\ 0.1\ 0.1\ 0.1]^{\intercal}$ and covariance $\Sigma_{\mathbf{x}_0} = \mathrm{diag}(1\mathrm{E}-6, 1\mathrm{E}-6, 1\mathrm{E}-8, 1\mathrm{E}-8,1\mathrm{E}-3,1\mathrm{E}-3,1\mathrm{E}-3)$ as well as a
Gaussian disturbance for the noise in angular velocity, $\mathbf{w} \in \mathbb{R}^{3}$, with mean $\mu_{\mathbf{w}_k} = {[0\ 0\ 0]}^\intercal$ and covariance matrix $\Sigma_{\mathbf{w}_k} = 5\mathrm{E}-3\times\mathrm{diag}(1,1,1)$.
For the initial condition and disturbance, we presume the $90\% $ level set.
The time horizon is $N=5$ with sampling time $T_s = 0.01$ to discretize via Euler discretization, 
\begin{equation}
    \mathbf{x}_{k+1} = \mathbf{x}_k + T_s f_k(\mathbf{x}_k,\mathbf{w}_k), 
\end{equation}
where $\mathbf{x}_k = [\mathbf{q}_0\ \mathbf{q}_1\ \mathbf{q}_1\ \mathbf{q}_3\ \boldsymbol{\omega}_1\ \boldsymbol{\omega}_2\ \boldsymbol{\omega}_3]^\intercal\in\R^7$ and $\mathbf{w}_k\in\R^3$. 
We presume that the initial state and the disturbance reside within intervals, \eqref{eq:intervalProbAlg}, with $90\%$ probability (i.e. $1-\delta_0 = 0.9$).
For sake of space, we do not present the decomposition function of this system and instead refer to the implementation the authors provide in~\cite{abate2021decomposition}.
\begin{figure}[ht!]
    \centering
    \includegraphics[width=\linewidth,trim={0.25cm 0.025cm 0.5cm 0.1cm}, clip]{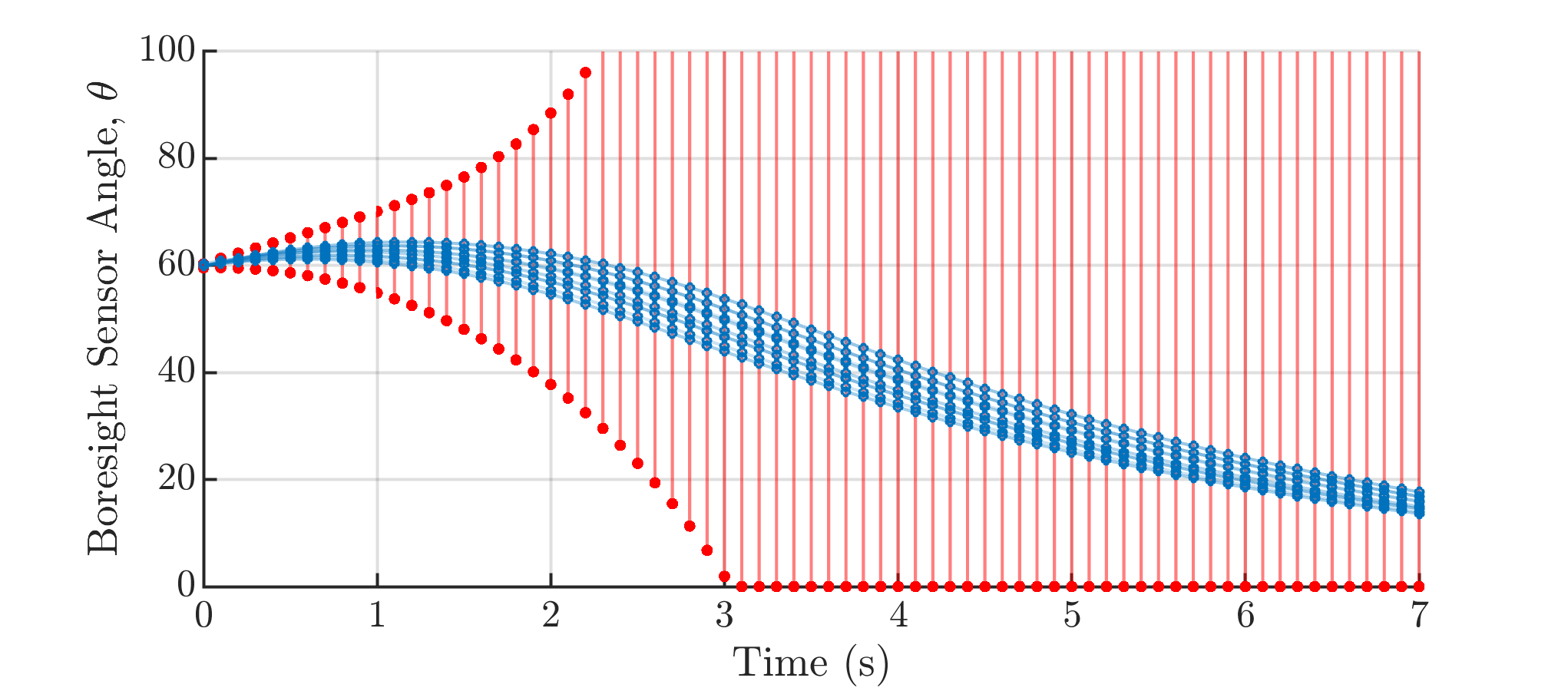}
    \caption{Resulting spacecraft trajectory subjected to random initial condition and random disturbances which reside in their initial intervals with $90\%$ probability (i.e. $1-\delta_0 = 0.9$), \eqref{eq:intervalProbAlg}. 
    In contrast to the linear example, the nonlinear system is highly sensitive to the underlying noise and results in stochastic interval reach sets to blow up to cover all possible values of \eqref{eq:angle_sensor}.}
    \label{fig:Example2F1}
\end{figure}
Figure~\ref{fig:Example2F1} shows sample trajectories with the interval reach sets via the mixed-monotone decomposition. 
Note that due to the system being nonlinear and highly sensitive to perturbations, the interval reach set blows up to the range of \eqref{eq:angle_sensor}.


\section{Conclusion and Future Directions}

The stochastic extension of the mixed-monotonicity property allows us, as in the deterministic case, to compute interval over-approximations of reachable sets through an efficient simulation-based algorithm.
We achieve this result by deriving stochastic order for a mixed monotone stochastic system, 
We establish the correctness guarantee for the algorithm using our stochastic extension of the mixed monotonicity property,
recovering a concentration bound on propagation of a distribution through a monotone mapping.
Future work will look to improving the concentration bound in order to mitigate the decrease in confidence over long time horizons.
Alternatively, we may modify the algorithm to enlarge the $X_k$ sets in order to maintain a constant confidence over time.



\section*{Acknowledgements}

We thank M. Abate and S. Coogan for making their code available for the 7D spacecraft example.


\bibliographystyle{IEEEtran}
\bibliography{refs}

\begin{thebibliography}{10}
\providecommand{\url}[1]{#1}
\csname url@samestyle\endcsname
\providecommand{\newblock}{\relax}
\providecommand{\bibinfo}[2]{#2}
\providecommand{\BIBentrySTDinterwordspacing}{\spaceskip=0pt\relax}
\providecommand{\BIBentryALTinterwordstretchfactor}{4}
\providecommand{\BIBentryALTinterwordspacing}{\spaceskip=\fontdimen2\font plus
\BIBentryALTinterwordstretchfactor\fontdimen3\font minus
  \fontdimen4\font\relax}
\providecommand{\BIBforeignlanguage}[2]{{%
\expandafter\ifx\csname l@#1\endcsname\relax
\typeout{** WARNING: IEEEtran.bst: No hyphenation pattern has been}%
\typeout{** loaded for the language `#1'. Using the pattern for}%
\typeout{** the default language instead.}%
\else
\language=\csname l@#1\endcsname
\fi
#2}}
\providecommand{\BIBdecl}{\relax}
\BIBdecl

\bibitem{Holzinger2021Cislunar}
M.~J. Holzinger, C.~C. Chow, and P.~Garretson, ``A primer on cislunar space,''
  Air Force Research Laboratory, Wright-Patterson Air Force Base, AFRL Report
  AFRL-2021-1271, 2021, approved for public release; distribution unlimited.

\bibitem{kvasnica2015reachability}
M.~Kvasnica, B.~Tak{\'a}cs, J.~Holaza, and D.~Ingole, ``Reachability analysis
  and control synthesis for uncertain linear systems in mpt,''
  \emph{IFAC-PapersOnLine}, vol.~48, no.~14, pp. 302--307, 2015.

\bibitem{kurzhanskiy2006ellipsoidal}
A.~A. Kurzhanskiy and P.~Varaiya, ``Ellipsoidal toolbox,'' in \emph{Proceedings
  of the 45th IEEE Conference on Decision and Control}.\hskip 1em plus 0.5em
  minus 0.4em\relax IEEE, 2006, pp. 1498--1503.

\bibitem{girard2005reachability}
A.~Girard, ``Reachability of uncertain linear systems using zonotopes,'' in
  \emph{International workshop on hybrid systems: Computation and
  control}.\hskip 1em plus 0.5em minus 0.4em\relax Springer, 2005, pp.
  291--305.

\bibitem{vinod2017forward}
A.~P. Vinod, B.~HomChaudhuri, and M.~M. Oishi, ``Forward stochastic
  reachability analysis for uncontrolled linear systems using fourier
  transforms,'' in \emph{Proceedings of the 20th International Conference on
  Hybrid Systems: Computation and Control}, 2017, pp. 35--44.

\bibitem{vinod2020probabilistic}
A.~P. Vinod and M.~M. Oishi, ``Probabilistic occupancy via forward stochastic
  reachability for markov jump affine systems,'' \emph{IEEE Transactions on
  Automatic Control}, vol.~66, no.~7, pp. 3068--3083, 2020.

\bibitem{sankaranarayanan2020reasoning}
S.~Sankaranarayanan, Y.~Chou, E.~Goubault, and S.~Putot, ``Reasoning about
  uncertainties in discrete-time dynamical systems using polynomial forms.''
  \emph{Advances in Neural Information Processing Systems}, vol.~33, pp.
  17\,502--17\,513, 2020.

\bibitem{stankovic1996taylor}
B.~Stankovi{\'c}, ``Taylor expansion for generalized functions,'' \emph{Journal
  of mathematical analysis and applications}, vol. 203, no.~1, pp. 31--37,
  1996.

\bibitem{estrada1993taylor}
R.~Estrada and R.~Kanwal, ``Taylor expansions for distributions,''
  \emph{Mathematical methods in the Applied Sciences}, vol.~16, no.~4, p. 297,
  1993.

\bibitem{xue2017just}
B.~Xue, M.~Fr{\"a}nzle, and P.~N. Mosaad, ``Just scratching the surface:
  Partial exploration of initial values in reach-set computation,'' in
  \emph{2017 IEEE 56th Annual Conference on Decision and Control (CDC)}.\hskip
  1em plus 0.5em minus 0.4em\relax IEEE, 2017, pp. 1769--1775.

\bibitem{xue2021reach}
B.~Xue, R.~Li, N.~Zhan, and M.~Fr{\"a}nzle, ``Reach-avoid analysis for
  stochastic discrete-time systems,'' in \emph{2021 American Control Conference
  (ACC)}.\hskip 1em plus 0.5em minus 0.4em\relax IEEE, 2021, pp. 4879--4885.

\bibitem{thorpe2021learning}
A.~J. Thorpe, K.~R. Ortiz, and M.~M. Oishi, ``Learning approximate forward
  reachable sets using separating kernels,'' in \emph{Learning for dynamics and
  control}.\hskip 1em plus 0.5em minus 0.4em\relax PMLR, 2021, pp. 201--212.

\bibitem{lew2021sampling}
T.~Lew and M.~Pavone, ``Sampling-based reachability analysis: A random set
  theory approach with adversarial sampling,'' in \emph{Conference on robot
  learning}.\hskip 1em plus 0.5em minus 0.4em\relax PMLR, 2021, pp. 2055--2070.

\bibitem{devonport2020estimating}
A.~Devonport and M.~Arcak, ``Estimating reachable sets with scenario
  optimization,'' in \emph{Learning for dynamics and control}.\hskip 1em plus
  0.5em minus 0.4em\relax PMLR, 2020, pp. 75--84.

\bibitem{devonport2023data}
A.~Devonport, F.~Yang, L.~El~Ghaoui, and M.~Arcak, ``Data-driven reachability
  and support estimation with christoffel functions,'' \emph{IEEE Transactions
  on Automatic Control}, 2023.

\bibitem{hashemi2023data}
N.~Hashemi, X.~Qin, L.~Lindemann, and J.~V. Deshmukh, ``Data-driven
  reachability analysis of stochastic dynamical systems with conformal
  inference,'' in \emph{2023 62nd IEEE Conference on Decision and Control
  (CDC)}.\hskip 1em plus 0.5em minus 0.4em\relax IEEE, 2023, pp. 3102--3109.

\bibitem{angeli2003monotone}
D.~Angeli and E.~D. Sontag, ``Monotone control systems,'' \emph{IEEE
  Transactions on Automatic Control}, vol.~48, no.~10, pp. 1684--1698, 2003.

\bibitem{coogan2016mixed}
S.~Coogan, M.~Arcak, and A.~A. Kurzhanskiy, ``Mixed monotonicity of partial
  first-in-first-out traffic flow models,'' in \emph{55th IEEE Conference on
  Decision and Control}, 2016, pp. 7611--7616.

\bibitem{enciso2006nonmonotone}
G.~A. Enciso, H.~L. Smith, and E.~D. Sontag, ``Nonmonotone systems decomposable
  into monotone systems with negative feedback,'' \emph{Journal of Differential
  Equations}, vol. 224, no.~1, pp. 205--227, 2006.

\bibitem{angeli2014small}
D.~Angeli, G.~A. Enciso, and E.~D. Sontag, ``A small-gain result for
  orthant-monotone systems under mixed feedback,'' \emph{Systems \& Control
  Letters}, vol.~68, pp. 9--19, 2014.

\bibitem{meyer2018sampled}
P.-J. Meyer, S.~Coogan, and M.~Arcak, ``Sampled-data reachability analysis
  using sensitivity and mixed-monotonicity,'' \emph{IEEE Control Systems
  Letters}, vol.~2, no.~4, pp. 761--766, 2018.

\bibitem{meyer2020ifac}
P.-J. Meyer and M.~Arcak, ``Interval reachability analysis using second-order
  sensitivity,'' in \emph{Proceedings of the $21^{st}$ IFAC World Congress},
  2020.

\bibitem{meyer2019tira}
P.-J. Meyer, A.~Devonport, and M.~Arcak, ``{TIRA}: toolbox for interval
  reachability analysis,'' in \emph{Proceedings of the 22nd ACM International
  Conference on Hybrid Systems: Computation and Control}.\hskip 1em plus 0.5em
  minus 0.4em\relax ACM, 2019, pp. 224--229.

\bibitem{meyer2021interval}
------, \emph{Interval Reachability Analysis: Bounding Trajectories of
  Uncertain Systems with Boxes for Control and Verification}.\hskip 1em plus
  0.5em minus 0.4em\relax Springer Nature, 2021.

\bibitem{coogan2020mixed}
S.~Coogan, ``Mixed monotonicity for reachability and safety in dynamical
  systems,'' in \emph{2020 59th IEEE Conference on Decision and Control
  (CDC)}.\hskip 1em plus 0.5em minus 0.4em\relax IEEE, 2020, pp. 5074--5085.

\bibitem{stocOrderShakedShantikumar}
M.~Shaked and J.~G. Shanthikumar, \emph{Stochastic orders}.\hskip 1em plus
  0.5em minus 0.4em\relax Springer, 2007.

\bibitem{coogan2015efficient}
S.~Coogan and M.~Arcak, ``Efficient finite abstraction of mixed monotone
  systems,'' in \emph{18th International Conference on Hybrid Systems:
  Computation and Control}, 2015, pp. 58--67.

\bibitem{dolecki2016convergence}
S.~Dolecki and F.~Mynard, \emph{Convergence foundations of topology}.\hskip 1em
  plus 0.5em minus 0.4em\relax World Scientific Publishing Company, 2016.

\bibitem{wiesel1989spaceflight}
W.~E. Weisel, \emph{Spaceflight dynamics}.\hskip 1em plus 0.5em minus
  0.4em\relax New York, McGraw-Hill Book Co, 1989, vol.~2.

\bibitem{lesser2013stochastic}
K.~Lesser, M.~Oishi, and R.~S. Erwin, ``Stochastic reachability for control of
  spacecraft relative motion,'' in \emph{52nd IEEE Conference on Decision and
  Control}.\hskip 1em plus 0.5em minus 0.4em\relax IEEE, 2013, pp. 4705--4712.

\bibitem{abate2021decomposition}
M.~Abate and S.~Coogan, ``Decomposition functions for interconnected mixed
  monotone systems,'' \emph{IEEE Control Systems Letters}, vol.~6, pp.
  2120--2125, 2021.

\end{thebibliography}
\end{document}